\documentclass[11pt]{article}
 
\usepackage[utf8]{inputenc}
\usepackage{authblk}

\usepackage{fullpage}
\usepackage{soul}
\usepackage{tabularx}
\usepackage{multirow} 
\usepackage{thm-restate}
\usepackage{amsthm}
\usepackage{amsmath}
\usepackage{amssymb}
\usepackage{amsfonts}
\usepackage{mathtools}
\usepackage{enumitem}
\usepackage[pagebackref]{hyperref}
\usepackage[sort,numbers]{natbib}
\usepackage[dvipsnames,svgnames,x11names]{xcolor}

\usepackage{tikz}

\usepackage[capitalize]{cleveref}
\usepackage{caption}
\usepackage{xspace}
\usepackage{multicol}

\usepackage{todonotes}

\newtheorem{theorem}{Theorem}

\newtheorem{corollary}[theorem]{Corollary}
\newtheorem{proposition}[theorem]{Proposition}

\theoremstyle{definition}

\crefname{claim}{Claim}{Claims}

\newcommand{\problemdef}[3]{
	\begin{center}\fbox{
	\begin{minipage}{150mm}
		\noindent
		#1
		\\
		\setlength{\tabcolsep}{3pt}
		\begin{tabularx}{\textwidth}{@{}lX@{}}
			\textrm{Input:}     & #2 \\
			\textrm{Question:}  & #3
		\end{tabularx}
	\end{minipage}}
	\end{center}
}

\usepackage{xspace}

\let\oldlambda\lambda
\renewcommand{\lambda}{\ensuremath{\oldlambda}\xspace}
\let\oldalpha\alpha
\renewcommand{\alpha}{\ensuremath{\oldalpha}\xspace}
\let\oldDelta\Delta
\renewcommand{\Delta}{\ensuremath{\oldDelta}\xspace}

\newcommand{\EE}{\ensuremath{\mathcal E}\xspace}

\newcommand{\GG}{\ensuremath{\mathcal G}\xspace}
\newcommand{\HH}{\ensuremath{\mathcal H}\xspace}

\renewcommand{\SS}{\ensuremath{\mathcal S}\xspace}
\newcommand{\TT}{\ensuremath{\mathcal T}\xspace}
\newcommand{\UU}{\ensuremath{\mathcal U}\xspace}
\newcommand{\VV}{\ensuremath{\mathcal V}\xspace}
\newcommand{\WW}{\ensuremath{\mathcal W}\xspace}

\newcommand{\eps}{\ensuremath{\varepsilon}\xspace}
\renewcommand{\epsilon}{\eps}

\newcommand{\ignore}[1]{}

\renewcommand{\leq}{\leqslant}
\renewcommand{\geq}{\geqslant}
\renewcommand{\ge}{\geqslant}
\renewcommand{\le}{\leqslant}

\newcommand{\CCVM}{{\sc Challenge the Champ Value Maximization}\xspace}
\newcommand{\CCVMD}{{\sc Challenge the Champ Value Maximization-Dag}\xspace}

\newcommand{\TDM}{{\sc 3-D-Matching}\xspace}

\newcommand{\Pb}{\ensuremath{\mathsf{P}}\xspace}
\newcommand{\NP}{\ensuremath{\mathsf{NP}}\xspace}
\newcommand{\NPC}{\ensuremath{\mathsf{NP}}\text{-complete}\xspace}
\newcommand{\NPH}{\ensuremath{\mathsf{NP}}\text{-hard}\xspace}

\newcommand{\el}{\ensuremath{\ell}\xspace}

\title{Maximizing Value in Challenge the Champ Tournaments}

\author[1]{Umang~Bhaskar\thanks{Supported by the DAE, Government of India, project no.~RTI4001. (email: umang@tifr.res.in).}}
\author[1]{Juhi~Chaudhary\thanks{Supported by the DAE, Government of India, project no.~RTI4001. (email: juhi.chaudhary@tifr.res.in).}}
\author[2]{Palash~Dey$^*$\thanks{Supported by the SERB, Government of India, grant no. CRG/2022/003294. (email:palash.dey@cse.iitkgp.ac.in).}}

\affil[1]{\small School of Technology and Computer Science, Tata Institute of Fundamental Research, Mumbai, 
India.}
\affil[2]{\small Department of Computer Science, Indian Institute of Technology Kharagpur, India.}

\date{}

\begin{document}

\maketitle

\begin{abstract}

A tournament is a method to decide the winner in a competition, and describes the overall sequence in which matches between the players are held. While deciding a worthy winner is the primary goal of a tournament, a close second is to maximize the value generated for the matches played, with value for a match measured either in terms of tickets sold, television viewership, advertising revenue, or other means. Tournament organizers often seed the players --- i.e., decide which matches are played --- to increase this value.

We study the value maximization objective in a particular tournament format called \emph{Challenge the Champ}. This is a simple tournament format where an ordering of the players is decided. The first player in this order is the initial champion. The remaining players in order challenge the current champion; if a challenger wins, she replaces the current champion. We model the outcome of a match between two players using a complete directed graph, called a \emph{strength graph}, with each player represented as a vertex, and the direction of an edge indicating the winner in a match. The value-maximization objective has been recently explored for knockout tournaments when the strength graph is a directed acyclic graph (DAG). 
We extend the investigation to Challenge the Champ tournaments and general strength graphs. We study different representations of the value of each match, and completely characterize the computational complexity of the problem. 
\bigskip

\end{abstract}

\section{Introduction}
Tournaments\footnote{In this paper, a \emph{tournament} will denote a competition format, and not a complete directed graph.} are a fundamental mechanism for determining winners in diverse competitive settings by systematically comparing participants through a series of pairwise matches. The most illustrative example of tournaments is sports competitions, ranging from prestigious international events like the Olympics and World Cups to local school sports and college contests. Different tournament formats, such as \emph{knockout} tournament \cite{suksompong2021tournaments,williams_moulin_2016}, \emph{round-robin} tournament \cite{russell2009manipulating,suksompong2016scheduling}, \emph{double-elimination tournament} \cite{aziz2018fixing,stanton2013structure}, \emph{Swiss-Sytem} \cite{sauer2024improving,sziklai2022efficacy} are commonly employed and have been widely studied by researchers over the years, where each format offers a distinct approach with its own advantages and disadvantages for ranking and elimination.

Challenge the Champ --- also known as stepladder --- is one such tournament format, which, despite its simplicity, is relatively underexplored in the literature compared to other tournament variants. Given $n$ players, a Challenge the Champ tournament proceeds in $n-1$ rounds. A player is initially chosen as the champ, and in each successive round, a new player challenges the current champ. If a player defeats the current champ, she becomes the new champ. The champ at the end of the final round is the tournament winner. See \cref{fig1} for an illustration. Challenge the champ is a \emph{single-elimination tournament}, where a player is eliminated once she loses a match. Thus, in a single-elimination tournament with $n$ players, there are exactly $n-1$ matches, though there are of course many ways to decide the player in each match \footnote{A \emph{knockout} tournament, which takes the form of a binary tree, is also sometimes referred to as a single-elimination tournament. However, we use the more general definition of single-elimination tournaments in this paper}.

Note that the order in which the players challenge the current champ --- called the \emph{seeding} --- significantly influences the outcome of the tournament. E.g., a fixed player $i$ has a better likelihood of winning the tournament if she challenges the champ late in the tournament after the stronger players have been eliminated, rather than very early.

Owing to their simplicity and often dramatic nature --- allowing a new entrant to beat the current champion and take over the title ---  Challenge the Champ tournaments and its variants are used in multiple sports, including ten-pin bowling, squash, badminton, and basketball~\cite{wiki,ArlegiD20}. Beyond the world of sports, a continuous variant of stepladder tournaments is also used for ranking workers in organisations~\cite{Pongou17}, making it a compelling tournament format for theoretical studies. In fact, some of our results obtained from studying Challenge the Champ tournaments are applicable to all single-elimination tournaments, further reinforcing their importance.

\begin{figure}[t]
 \centering
    \includegraphics[scale=.6]{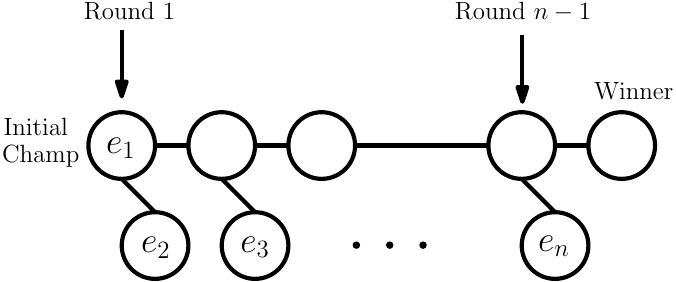}
    \caption{A Challenge the Champ tournament, with players seeded $(e_1,e_{2},\ldots,e_{n})$. In each round $i$, player $e_{i+1}$ challenges the current champ --- the winner of the previous round.}
    \label{fig1}
\end{figure}

\newcommand{\mrrb}[2]{\multirow{#1}{2.4cm}{\centering #2}}
\renewcommand{\arraystretch}{1.5}
\newcolumntype{Y}{>{\centering\arraybackslash}X}
\renewcommand\tabularxcolumn[1]{m{#1}}
\captionsetup[table]{skip=6pt}
\begin{table*}[t]
  \setlength{\tabcolsep}{8pt}
  \centering
  \caption
  {Our results for value maximization in Challenge the Champ tournaments.}
  \def\fbs{3.5cm}
\small
  \begin{tabularx}{.95\textwidth}{|@{}Y|Y|Y|Y@{}|}%
 \hline
      \textbf{Problem} & \textbf{Pair-Based} & \textbf{Win-Count-Based} & \textbf{Player-Popularity-Based}	\\
      \hline
      \textsc{CTC-VM-Dag} & \NPC for binary values (\cref{thm:dag:npc}) & \textsf{P} (\cref{thm:pp-alg2}) & \textsf{P} (\cref{thm:pp-alg1}) \\
       \hline
       \multirow{2}{*}{\centering \textsc{CTC-VM}} &  \multirow{2}{1.5in}{\centering \NPC(Follows from \cref{thm:dag:npc})} & \multirow{2}{1.5in}{\centering \NPC for binary values (\cref{thm:ccvm-nph-bin}) and linear-after-threshold (\cref{thm:ccvm-nph-lin})}  & \textsf{P} for binary values (\cref{thm:pp-alg}) \\\cline{4-4}
       &  &  & \NPC for ternary values (\cref{thm:rb-np-com}) \\
  \hline
  \end{tabularx}
  \label{table:challenge-the-champ}
\end{table*}

Challenge the Champ tournaments have previously attracted some attention from researchers, though this has been focused on the complexity of manipulation to ensure a given player wins the tournament~\cite{mattei2015complexity,chaudhary2024}.
We are, however, interested in the objective of value maximization in sports tournaments. Arguably, many competitions are organized to maximize some measure of value --- whether advertising revenue, viewership, or attendance. Given a tournament format --- such as knockout, or Swiss-system --- often this value-maximization objective influences the matches that are played as well, e.g., through the choice of groups, or seeding of the players. This objective is measured by assigning a non-negative integer value to each potential match between the players. Such an approach reflects the practical scenario where matches hold varying levels of importance, influenced by factors like geopolitics, historical rivalries, intrinsic fan following, or the relative strength of the teams. Prioritizing high-value matches becomes essential, especially when only a subset of potential matches can be played, as it helps attract viewers and advertisers. The \emph{tournament value} is then defined as the total sum of the values of all matches that are actually played. This metric is crucial for evaluating and optimizing the overall appeal and financial success of the tournament. We study the problem of value maximization in Challenge the Champ tournaments. Given the expanding body of research on tournaments, we believe that studying value maximization in the simple but interesting Challenge the Champ tournament format is an important research problem.

To complete the picture, we need a way of determining the outcome when two players compete. For this, we use a \emph{strength graph}, which is a complete directed graph with a vertex for each player. A directed edge $(i,j)$ indicates that player $i$ beats player $j$ in a match. Although value maximization has been previously examined in knockout tournaments \cite{gupta2024exercise,chaudhary2024make}, the investigation is restricted to the case where the strength graph is a directed acyclic graph (DAG). Our study extends this inquiry to Challenge the Champ tournaments. We investigate scenarios where the strength graph is either a DAG or a complete directed graph, and provide a thorough analysis of the computational complexity of value maximization in Challenge the Champ tournaments.\\

\noindent{\textbf{Related Work.}} 
Challenge the Champ tournaments are also sometimes called stepladder tournaments in previous work. These have been studied from the perspective of satisfying axiomatic notions of fairness~\cite{ArlegiD20}, as well as the characterizing strength graphs where a favorite player can win the tournament~\cite{knockout,YangD21}.~\citet{ArlegiD20} highlight instances of stepladder tournaments being utilized in sports competitions, such as ten-pin bowling and squash. Stepladder tournaments are also useful in firms.~\citet{Pongou17} present examples, and discuss the relation of the ranking of workers obtained from stepladder tournaments with their importance in the firm.

For Challenge the Champ tournaments with probabilistic strength graphs (rather than deterministic, as in our case),~ \citet{mattei2015complexity} investigated a setting where players can be bribed to lower their winning probability against
the initial champ. The goal was to maximize the probability of the initial champ winning the tournament
by bribing the other players while not exceeding a given budget for the bribes. Building on this,~\citet{chaudhary2024} extended the research by examining the problem with respect to various parameters. 

For tournament value maximization,~\citet{chaudhary2024make} initiated this study for knockout tournaments under the name \textsc{Tournament Value Maximization} in a deterministic setting, where players have a strict strength ordering (the strength graph forms a DAG).  Their work explored various constraints on tournament value functions to optimize the overall tournament value. In a related study,~\citet{gupta2024exercise} examined a similar problem but without assuming a strict linear order of player strengths. Their focus was on determining whether a specific seeding could guarantee that a designated set of games, known as demand matches, would occur. This problem can be viewed as a special case of \textsc{Tournament Value Maximization}, where demand matches are assigned a value of 1, and other matches a value of 0, with the objective of ensuring the total value meets or exceeds the number of demand matches.
Additionally,~\citet{dagaev2018competitive} investigated a restricted variant of tournament value maximization, where each player has a unique strength value, and the value of a game is defined by a linear combination of its ``quality" (the sum of the players' strengths) and its ``intensity" (the absolute difference between the players' strengths). They characterized scenarios in which either a ``close” seeding, a ``distant” seeding, or any seeding could be optimal, demonstrating that their restricted cases can be solved efficiently in linear time.

\section{Problem Statement and Preliminaries}
 For a positive integer $n$, define $[n] := \{1, 2, \ldots, n\}$. An instance of \CCVM is given by a set of players $N = [n]$, a strength graph, which is a complete directed graph $\TT = (N,E)$ with a directed edge $(i,j)$ indicating that player $i$ beats player $j$ in a match, and a value function $v$ for each possible match in the tournament. We present results for both the case when $\TT$ is a DAG and when $\TT$ contains directed cycles. If $\TT$ is a DAG, we will assume all edges are oriented from larger indices to smaller indices. Hence $n$ beats every other player, while $1$ is beaten by every other player. If a player $i$ defeats a player $j$, we say that $i$ is \emph{stronger} than $j$, or equivalently, $j$ is \emph{weaker} than $i$.

Given a \CCVM instance with $n$ players, a seeding $\sigma \in \mathbb{S}_n$ (where $\mathbb{S}_n$ is the symmetric group on $n$ elements) is a permutation of the set of players that completely determines the order of the matches. The player $\sigma(1)$ is the initial champ. In each round $r=1, \ldots, n-1$, the champ is challenged by player $\sigma(r+1)$. The winner of the match (as determined by the strength graph $\TT$) is the champ for the next time step. The winner in the last round is the winner of the tournament. Thus, given $N$ and $\TT$, each seeding $\sigma$ completely determines the set of matches that take place. Given a seeding $\sigma$, we define $M_i(\sigma)$ as the set of players that play a match with player $i$ and lose to it, and define $w_i(\sigma) = |M_i(\sigma)|$. Note that player $i$ wins exactly $w_i(\sigma)$ matches.

We consider a number of different possible settings for the tournament value $V(\sigma)$ of a seeding $\sigma$ for a Challenge the Champ tournament. \\

\noindent{\textbf{Player-popularity-based.}} We first consider a natural restriction when each player has an associated popularity $p_i \in \mathbb{{Z}_+}$, and the value of each match is determined by the popularity of the winning player. Thus, each match player $i$ wins contributes value $p_i$. Thus the value for the seeding $\sigma$ is $V(\sigma) = \sum_{i \in N} p_i w_i(\sigma)$, and such a value function is called \emph{player-popularity-based}. This restriction captures the intrinsic popularity of certain teams and players among their fans, irrespective of the opponent or venue. This variant has also been explored by \citet{chaudhary2024make}.

Mathematically, when the strength graph is a DAG, the player-popularity-based value function is \emph{equivalent} to the tournament value function where the value of every match is defined as the sum of the popularity values of the two players participating in that match. This equivalence holds for any single elimination tournament (not necessarily Challenge the Champ), if the strength graph is a DAG. The equivalence can be seen easily: Let $V'(\sigma)$ denote the tournament value when the value of the match between players $i$ and $j$ is $v'(i,j)=p_{i}+p_{j}$, and let $V(\sigma)$ denote the tournament value when the value of the match is $v(i,j)=p_{i}$, where $i>j$. Then $V'(\sigma)=V(\sigma)+\sum_{i\in[n-1]}p_{i}$, since every player except player $n$ loses exactly one match in any single elimination tournament. Thus, when the match values are modified, all players except the winner additionally contribute their popularity value exactly once, for the match where they lose. Thus, since the two tournament value functions are equivalent for DAGs, we will present results for the player-popularity-based value function. \\


\noindent{\textbf{Win-count-based.}} We next consider tournament value functions where the value of each match depends not only on the winning player, but also by her track record of victories. Specifically, the value of a match increases with the number of wins the winning player has accumulated. We refer to such value functions as \emph{win-count-based} value functions. Formally, each player has a win-value $f_i: [n-1] \rightarrow \mathbb{Z}_+$ that gives the value of each match won by player $i$. I.e., the $k$th match won by player $i$ has value $f_i(k)$. Then the value for the seeding $\sigma$ for win-count-based value function is $V(\sigma) = \sum_{i \in N} \sum_{k=1}^{w_i(\sigma)} f_i(k)$. Note that if each player has a constant win-value function , i.e., if $f_i(k) = p_i$ for every $k\in [w_{i}(\sigma)]$, this is the same as player-popularity-based value functions.

Within win-count-based values, we consider two further restrictions. We say a win-count-based tournament value function is \emph{binary-valued} if every player $i$ has a threshold $\lambda_i$, and the win-value function $f_i(x) = 1$ if $x = \lambda_i$, and is zero otherwise. Thus, player $i$ gets a value of 1 if she wins at least $\lambda_i$ games, and gets value 0 otherwise. The tournament value function is thus $V(\sigma) = |\{i: w_i(\sigma) \ge \lambda_i\}|$. These value functions support the egalitarian objective of ensuring that more players/teams achieve their individual thresholds. Maximizing the tournament value, in this case, is the same as obtaining a seeding that enables as many players as possible to meet their thresholds, rather than allowing a single player to dominate by participating in the majority of matches.

We say a win-count-based tournament value function is \emph{linear-after-threshold} if every player $i$ has a threshold $\lambda_i$, and the win-value function $f_i(x) = 1$ if $x \ge \lambda_i$, and is zero otherwise. Thus, player $i$ gets a value of 1 for every game she wins after $\lambda_i-1$, and the tournament value function is $V(\sigma) = \sum_{i \in N} \max\{0, w_i(\sigma) - \lambda_i + 1\}$. \\

\noindent{\textbf{Pair-based.}} Finally, we consider \emph{pair-based} value function, where the value of each match depends on both the players in the match. Thus, such value functions account for matches where the players have a historic rivalry or are particularly competitive. Here, each pair of players $i$, $j$ has a pair value $f_{ij} \in \mathbb{Z}_+$, which is the value contributed by a match between $i$ and $j$. The value of a seeding $\sigma$ is $V(\sigma) = \sum_{i \in N} \sum_{j \in M_i(\sigma)} f_{ij}$. 

Again, pair-based value functions generalize player-popularity-based value functions. To see this, if each player $i$ has popularity $p_i$, consider the pair-based value function where the pair $\{i,j\}$ has value $p_i$ if $i$ beats $j$, and $p_j$ otherwise. Clearly, each match then has the same value in either case.

These value functions have been studied before for value maximization in knockout tournaments~\cite{chaudhary2024make}. Pair-based value functions are also called \emph{round-oblivious} value functions. We will further consider restrictions on the value of each match to be binary or ternary (i.e., the value of each match is in $\{0,1\}$ or in $\{0,1,2\}$).

More formally, the definition of \CCVM is given below.
\problemdef{{\CCVM \textsc(CTC-VM)}}{
Given a set $N$ of players, a strength graph $\mathcal{T}$ on $N$ which is a complete directed graph, a tournament value function $V$, and a target value $t$.
}{
Is there a tournament seeding $\sigma$ for the players in $N$ such that the tournament value $V(\sigma) \ge t$?
}

The problem \CCVMD \textsc{(CTC-VM-Dag)} is defined in a similar manner, except that the input strength graph, in this case, is a DAG.\\

\noindent{\textbf{Seedings and Caterpillars.}} For undirected graphs, a \emph{caterpillar graph} is a tree in which the removal of all pendant vertices (leaves) yields a path, which we refer to as the \emph{backbone}. We extend this to directed graphs: a \emph{caterpillar arborescence} is an arborescence where removing all leaves (vertices with out-degree zero) yields a directed path. We call this path the \emph{backbone} of the caterpillar. A \emph{spanning caterpillar arborescence} of a directed graph is a subgraph that is a caterpillar arborescence on all the vertices. For brevity, since we will only consider spanning caterpillar arborescences on the directed graph $\TT$, we omit the term arborescence from the description.

Any seeding $\sigma$ in a Challenge the Champ tournament induces a spanning caterpillar in the strength graph, and vice versa. To see this, given a seeding $\sigma$, let the backbone of the caterpillar be the players that win at least one match, in the order given by $\sigma$. All remaining players lose their first match to a player in the backbone. Attach the remaining players to the player they lose to in the backbone. This gives a spanning caterpillar. For the converse, given a spanning caterpillar, the seeding is obtained by the sequence of players in the backbone, interspersed by the leaves attached to each player in the backbone in any order. Figure~\ref{fig:seedingsca} gives an example showing this correspondence. Thus, instead of a seeding, we can equivalently specify the sequence of players in the backbone of a caterpillar, and the attachment of the remaining players (i.e., the leaves) to the players in the backbone. We will frequently use this correspondence in our proofs. Corollary~\ref{corr:sc} also establishes a connection between pair-based value functions and spanning caterpillars. We note that there is prior work linking spanning caterpillars and tournaments (see, e.g., \cite{reid1983embedding,reid1991majority}). Our work significantly strengthens this connection by demonstrating that caterpillars provide a natural framework for understanding Challenge the Champ tournaments and establishing hardness results for caterpillars.

\begin{figure}[ht]
\centering
\begin{tikzpicture}[baseline=-2cm, 
  every node/.style={circle, draw, minimum size=4mm, inner sep=0pt},
  level distance=3.5mm,
  node distance=3.5mm,
  edge from parent/.style={draw, -latex},
  edge from parent path={(\tikzparentnode) -- (\tikzchildnode)}
  ]

  \node (v1) {3};
  \node (v2) [right=of v1] {4};
  \node (v3) [right=of v2] {4};
  \node (v4) [right=of v3] {5};
  \node (v5) [right=of v4] {7};
  \node (v6) [right=of v5] {7};
  \node (v7) [right=of v6] {7};

  \node (v1b) [below=of v1, xshift=2.5mm] {4};
  \node (v2b) [below=of v2, xshift=2.5mm] {2};
  \node (v3b)  [below=of v3, xshift=2.5mm] {5};
  \node (v4b)  [below=of v4, xshift=2.5mm] {7};
  \node (v5b)  [below=of v5, xshift=2.5mm] {1};
  \node (v6b)  [below=of v6, xshift=2.5mm] {6};

  \draw (v1) -- (v2) -- (v3) -- (v4) -- (v5) -- (v6) -- (v7);
  \draw (v1) -- (v1b);
  \draw (v2) -- (v2b);
  \draw (v3) -- (v3b);
  \draw (v4) -- (v4b);
  \draw (v5) -- (v5b);
  \draw (v6) -- (v6b);

\end{tikzpicture}
\qquad 
\begin{tikzpicture}[
  every node/.style={circle, draw, minimum size=4mm, inner sep=0pt},
  node distance=8mm, 
  level distance=10mm, 
  edge from parent/.style={draw, -latex},
  edge from parent path={(\tikzparentnode) -- (\tikzchildnode)}
  ]

  \node (v1) {7}
    child {node (v6) [below left of=v1, xshift=-2mm, yshift=-2mm] {1}}
    child {node (v7) [below right of=v1, xshift=2mm, yshift=-2mm] {6}}    
    child {node (v2) [below of=v1, xshift=0mm, yshift=-2mm] {5}
      child {node (v3) {4}
        child {node (v4) [below left of=v3, xshift=-2mm, yshift=-2mm] {2}}
        child {node (v5) [below right of=v3, xshift=2mm, yshift=-2mm] {3}}
      }
    };

  

\end{tikzpicture}
\caption{The first figure gives an example with 7 players showing the seeding $(3,4,2,5,7,1,6)$. In the strength graph (not shown), each player $j$ defeats all players $i < j$. Hence player 7 is the strongest and player 1 the weakest. The second figure shows the resulting spanning caterpillar arborescence.}

\label{fig:seedingsca}
\end{figure}
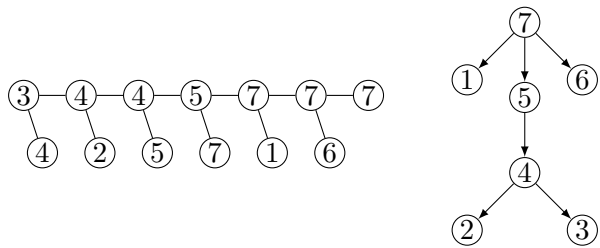

We will also use the following result in our proofs.

 \begin{proposition}[\cite{redei1934kombinatorischer}] \label{prop:ham}
    Given a complete directed graph $G$, a Hamiltonian path in $G$ can be found in polynomial time.
\end{proposition}

 \noindent{\textbf{Our Contribution.}} We present a comprehensive analysis of the complexity of tournament value maximization for Challenge the Champ tournaments. Our work encompasses both \textsc{CTC-VM} and \textsc{CTC-VM-Dag} across various tournament value functions, including player-popularity-based, win-count-based, and pair-based tournament value functions, and restrictions to binary and ternary-valued functions. For a tabular view of the main results, see \cref{table:challenge-the-champ}.  
 
We start by delving into player-popularity-based tournament value functions in \cref{sec:ppb}. We first establish that when the strength graph is a DAG, the tournament value can be maximized in polynomial time. Further, in this setting, Challenge the Champ tournaments are optimal among all single-elimination tournaments. Specifically, we show:
\begin{itemize}
    \item For every strength graph with player-popularity-based tournament value functions, there exists a Challenge the Champ tournament that maximizes the total value over all single-elimination tournaments. Moreover, this maximum value and the seeding can be computed in polynomial time.
\end{itemize}

We note that in knockout tournaments, this problem is unresolved~\cite{chaudhary2024make}. Our proof, however, shows that the optimal value obtained for Challenge the Champ tournaments is an upper bound for all single-elimination tournaments, including knockout tournaments. 
 
 When the strength graph is no longer a DAG, the complexity of the problem varies with the number of distinct player popularity values. For binary values for the popularity of players, we give a polynomial-time algorithm, and observe that in this setting again, Challenge the Champ tournaments are optimal among all single-elimination tournaments. However, when the popularity of the players can take three values $\{0, 1, 2\}$, the problem becomes \NPH, \emph{even when there is a single player with popularity 2,} which implies that the problem is para-NP-hard when parameterized by the number of players with popularity value 2.
 \begin{itemize}
     \item  \textsc{CTC-VM} is polynomial-time solvable for player-popularity-based tournament value functions when each player's popularity is in $\{0,1\}$.
 \end{itemize}

 \begin{itemize}
     \item \textsc{CTC-VM} is \textsf{NP}-complete for player-popularity-based tournament value functions when each player's popularity is in  $\{0,1,2\}$.
 \end{itemize}

We also present a simple greedy approximation algorithm for \textsc{CTC-VM}, where each player's popularity is drawn from a set of $k$ distinct values by leveraging the result of \textsc{CTC-VM} for player-popularity-based tournament value functions where each player's popularity is restricted to $\{0,1\}$. 

 In \cref{sec:wcb}, we explore the more general case of win-count-based tournament value functions. 
We provide a polynomial-time dynamic programming algorithm for the case where the input is a DAG. However, for general strength graphs, the problem becomes hard even if the win-value function $f_i$ for each player is binary. This highlights that the presence of directed cycles is a crucial factor in the complexity of the problem.

 \begin{itemize}
     \item \textsc{CTC-VM-Dag} is polynomial-time solvable for win-count-based tournament value functions.
 \end{itemize}

\begin{itemize}
    \item \textsc{CTC-VM} is \textsf{NP}-complete even for binary-valued and linear-after-threshold win-count-based tournament value functions.
\end{itemize}

The last result also highlights a difference from player-popularity-based functions: while the latter becomes complex with ternary values, the former is challenging even with binary values.


In \cref{sec:ggv}, we investigate pair-based tournament value functions and prove that the problem is \textsf{NP}-complete even for the simplest setting when the strength graph is a DAG, and each pair value is binary. Specifically, we prove:

\begin{itemize}
    \item \textsc{CTC-VM-Dag} is \textsf{NP}-complete for pair-based tournament value functions when all pair values are binary.
\end{itemize}

We find this surprising, given the relative simplicity of Challenge the Champ tournaments compared to other tournament formats. Finally, we show that this last result also implies hardness for a problem of value maximization for spanning caterpillars.

\section{Player-Popularity-Based Tournament Value Functions} \label{sec:ppb}

We first demonstrate that for player-popularity-based tournament value functions with acyclic strength graphs, the optimal tournament value can be obtained in polynomial time.

    \begin{theorem}
         There is a polynomial-time algorithm for  \textsc{CTC-VM-Dag} for player-popularity-based tournament value functions.
         \label{thm:pp-alg1}
    \end{theorem}
    \begin{proof}
        As stated, since $\TT$ is a DAG, we assume player $i \in N$ beats all players $j < i$, and player $i$ has popularity $p_i$. We now describe a greedy algorithm that constructs an optimal spanning caterpillar. We assume $p_i \neq p_j$ for any two players $i$, $j$, else we can slightly perturb the popularity values. Now let player $i_1$ be the most popular player. Let $W_1 := [i_1-1]$ be all the players weaker than $i_1$. Let player $i_2$ be the most popular player in $N \setminus [i_1]$, and let $W_2 := [i_2-1] \setminus [i_1]$. Thus recursively, we define $i_j$ to be the most popular player in $N \setminus [i_{j-1}]$, and $W_j := [i_j-1] \setminus [i_{j-1}]$. We continue until we pick player $n$, i.e., for some $k$, $i_k = n$. Note that each player in $W_j$ is less popular and weaker than player $i_j$, and player $i_j$ is stronger than player $i_{j-1}$.
        
        The spanning caterpillar has players $i_k$, $i_{k-1}$, $\ldots$, $i_1$ (in this order) on the backbone, and each $i_j$ additionally has edges to all the players in $W_j$. Each player $i_j$ in the backbone (other than $i_1$) thus wins $|W_j|+1$ matches, while player $i_1$ wins $|W_1|$ matches. The total value is thus
        \[ (i_1-1)\cdot p_{i_1} + (i_2-i_1)\cdot p_{i_2} + \ldots + (n-i_{k-1})\cdot p_n.\] \, 

        We now show that this is, in fact, the maximum value obtainable in any seeding. To see this, consider any other seeding $\sigma$. Then there are $n-1$ matches played, and there are $n-1$ values obtained from these matches. Consider the $t$th largest value obtained, say $v^*$, for any $t \in [n-1]$. We will show that this value is at most the $t$th largest value obtained by our algorithm, completing the proof.

        Let $j \in [k]$ be the largest index so that $t \ge i_j$. 
        Note that $i_{j+1}$ is the most popular player in $N \setminus [i_j]$. Then it can be checked that the $t$th largest value obtained by our algorithm is $p_{i_{j+1}}$. Now assume for a contradiction that $v^* > p_{i_{j+1}}$. Since $v^*$ is the $t$th largest value in the seeding $\sigma$, players with popularity at least $v^*$ must have won at least $t$ matches in total. But all players with popularity greater than $i_{j+1}$ lie in $[i_j]$. Hence these players can win at most $i_j - 1$, which is strictly less than $t$. This gives a contradiction.
\end{proof}

As noted, the proof of Theorem~\ref{thm:pp-alg1} shows something stronger. 
Any single-elimination tournament consists of at most $n-1$ matches, since in every match, one player is eliminated. Our proof shows that if the strength graph is a DAG, and the value of every match is the popularity of the stronger player, then the Challenge the Champ tournament, obtains maximum value over all possible single-elimination tournaments, and this value (and the seeding that obtains it) can be computed in polynomial time. To see this, it is enough to consider in the proof the $t$th largest value obtained from any single-elimination tournament; the proof shows that this is at most the value obtained by the seeding given for the Challenge the Champ tournament.

However, this is not the case when the strength graph is not necessarily a DAG. Even for player-popularity-based tournament value functions, there are instances where different tournament trees can be optimal. Figure~\ref{fig:threebranches} gives such an example. For \emph{binary} player-popularity values, however, Challenge the Champ tournaments can again be shown to be optimal for arbitrary strength graphs among all single-elimination tournaments.

\begin{figure}[ht]
\centering
\begin{tikzpicture}[
  every node/.style={circle, draw, minimum size=5mm, inner sep=0pt},
  node distance=8mm,
  level distance=9mm,
  level 1/.style={sibling distance=16mm},
  level 2/.style={sibling distance=8mm},
  edge from parent/.style={draw, -latex},
  edge from parent path={(\tikzparentnode) -- (\tikzchildnode)}
  ]

  \node {$r$}
    child {node {$a_1$}
      child {node {$b_1$}}
    }
    child {node {$a_2$}
      child {node {$b_2$}}
    }
    child {node {$a_3$}
      child {node {$b_3$}}
    };
\end{tikzpicture}
\caption{An example of a (partial) strength graph with cycles where a Challenge the Champ tournament is not optimal. Each leaf $b_i$ beats all the non-leaf nodes except $a_{i}$, creating cycles. The root has popularity $2$, intermediate $a_i$ vertices have popularity $1$, and the leaves have popularity $0$. Here, the maximum value achievable by a single elimination tournament is $9$ whereas the maximum value achievable in a Challenge the Champ tournament is $7$.}
\label{fig:threebranches}
\end{figure}
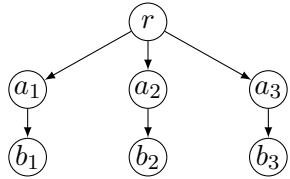

Next, we show that \textsc{CTC-VM} can be solved in polynomial time for player-popularity-based tournament value functions when each player's popularity is either 0 or 1. This result is significant because here, as for DAGs, not only do we get a polynomial time algorithm, but the value obtained is the largest among all single-elimination tournaments. Further, for general strength graphs, this is the largest class of valuations for which we obtain positive complexity results. For ternary popularity values, the problem becomes \NPH, as we show later.

\begin{restatable}{theorem}{THMD}\label{thm:pp-alg}
    There is a polynomial-time algorithm for  \textsc{CTC-VM} for player-popularity-based tournament value functions when each player's popularity is either 0 or 1. The value obtained is the maximum among all single-elimination tournaments.
\end{restatable}

\begin{proof}
Let $(N,\TT,V,t)$ be a given instance of \textsc{CTC-VM}. Since the popularity value of every player maps only to $\{0,1\}$, we categorize players with a popularity value of $1$ as \emph{popular} and those with a popularity value of $0$ as \emph{unpopular}. Next, we partition the players in $N$ into three sets as follows: 
\begin{itemize}
    \item Popular players, denoted by $\mathcal{P}$.
    \item Unpopular players who defeat all popular players, denoted by $\mathcal{W}$.
    \item Unpopular players who are defeated by at least one popular player, denoted by $\mathcal{U}$.
    
\end{itemize}

First, we show that the maximum achievable tournament value is $|\mathcal{P}| + |\mathcal{U}| - 1$, and this in fact holds for all single-elimination tournaments. Second, we provide a seeding strategy that achieves this maximum value.

For the first part, note that in any single-elimination tournament, matches involving players in $\mathcal{W}$ do not contribute to the tournament value, no matter who wins or loses. Matches where both players are in $\mathcal{U}$ similarly contribute nothing. A match between a player in $\mathcal{U}$ and a player in $\mathcal{P}$ contributes 1 if the player in $\mathcal{U}$ loses, hence these matches contribute at most $|\mathcal{U}|$. Matches between players in $\mathcal{P}$ can contribute a maximum of $|\mathcal{P}| - 1$ to the total tournament value. Consequently, the tournament value is bounded by $|\mathcal{P}| + |\mathcal{U}| - 1$.

For the second part, note that by \cref{prop:ham}, there exists a Hamiltonian path in both $G[\mathcal{P}]$\footnote{For a vertex subset $S \subseteq \mathcal{V}$ in a graph $\mathcal{G} = (\mathcal{V}, \mathcal{E})$, we use $G[S]$ to denote the subgraph of $G$ induced by $S$.} and $G[\mathcal{W}]$. Let $P$ denote such a Hamiltonian path in $G[\mathcal{P}]$, and let $p^*$ be the starting node of $P$. Let $W$ denote a Hamiltonian path in $G[\mathcal{W}]$, and let $w^*$ be the ending node of $W$. See \cref{fig5} for an illustration. Now, consider the Hamiltonian paths $W$ and $P$ as forming the backbone of a spanning caterpillar. This arrangement is valid since $w^{*}$ beats $p^{*}$ by definition of $\mathcal{W}$. Next, incorporate the vertices in $\mathcal{U}$ by attaching them as leaves to the vertex in $\mathcal{P}$ that beats them in the backbone. The resulting graph forms a spanning caterpillar with a tournament value of  $|\mathcal{P}|+|\mathcal{U}|-1$.  
\end{proof}

\begin{figure}[t]
 \centering
    \includegraphics[scale=.6]{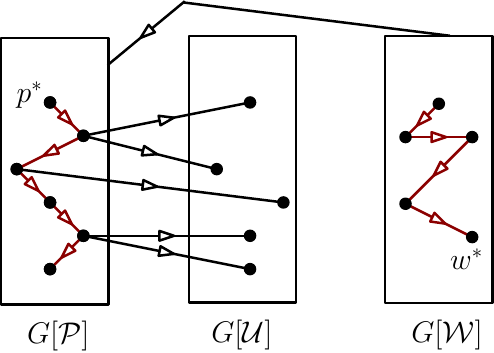}
    \caption{The red paths illustrate the Hamiltonian paths. A directed edge between two boxes represents all directed edges connecting the vertices from one box to those in the other.}
    \label{fig5}
\end{figure}
\medskip

Next, we prove that \textsc{CTC-VM} is \NPH when each player's popularity lies in $\{0,1,2\}$. Surprisingly, our hardness result holds even if there is a single player with popularity value 2. This gives a very precise threshold, as we have seen that for binary popularity values, the tournament value can be maximized in polynomial time.

We give a reduction from \TDM. The input of \TDM consists of a finite set $\mathcal{U} = X \cup Y \cup Z$ where $|X| = |Y| = |Z| = n$, and a collection $\mathcal{S} = \{S_1,  \ldots, S_m\}$ of triples, where each $S_i \subseteq X \times Y \times Z$. The problem is to determine if there exists a subset $\mathcal{S'} \subseteq \mathcal{S}$ consisting of $n$ triples such that each element in $\mathcal{U}$ appears in exactly one triple in $\mathcal{S'}$~\cite{garey1979computers}.

We will use the following proposition for \TDM.

\begin{restatable}{proposition}{PROP}
    For an instance of \TDM, if a family of subsets $\SS' \subseteq \SS$ covers elements $\UU' \subseteq \UU$, then $|\UU'| - |\SS'| \le 2n$, with equality iff $|\UU'| = 3n$ and $|\SS'| = n$. Hence if the given instance is a NO instance, then the inequality is always strict.
\label{prop:3dminequality}
\end{restatable}
\begin{proof}
Let $k = |\SS'|$ and $\ell = |\UU'|$. If $k < n$, then $\ell \le 3k$ since each set contains 3 elements, and $\ell - k \le 2k < 2n$. If $k > n$, then $\ell \le 3n$, and again $\ell - k \le 3n -k < 2n$. Thus if $\ell - k = 2n$, then $k = n$ and $\ell = 3n$.
\end{proof}

\begin{theorem}\label{thm:rb-np-com}
   \textsc{CTC-VM} is \textsf{NP}-complete for player-popularity-based tournament value functions when each player's popularity is either 0, 1, or 2.
\end{theorem}
\begin{proof}

The reduction from \TDM is as follows. We create a strength graph $\mathcal{T}$ with vertex set $N$. Let $H, M,$ and $L$ denote a partition of $N$ that corresponds to players with popularity 2, 1, and 0, respectively.
Corresponding to every element $a\in \mathcal{U}$, introduce a player $v_{a}$ in $L$. For each set $S\in \mathcal{S}$, introduce a player $v_S$ in $M$. If $a\in S$ for some $a\in \mathcal{U}$ and $S\in \mathcal{S}$, then the edge between $v_S$ and $v_a$ is directed from $v_S$ to $v_a$, else it is directed from $v_a$ to $v_S$. We introduce a new player $h$ in $H$ that beats every player in $M$ and loses to every player in $L$. The remaining edges in $\mathcal{T}$ are directed arbitrarily. We set the target value $t$ as 
$2(|\mathcal{S}|-n)+3n+2+(n-1)=2|\mathcal{S}|+2n+1$. The intuition behind the target value is that since each player in $M$ must choose between contributing to the tournament value by losing to $h$ (and thereby adding 2 points to the tournament value) or by defeating three players in $L$. Note that the factor $n-1$ arises because the players in $M$ who are part of the backbone can contribute additional points by playing among themselves. Also, the extra 2 is due to the fact that, at most, one player in the backbone that comes from $M$ can contribute to the tournament value by both beating some players and losing to $h$.

We need to show that $(\mathcal{U},\mathcal{S})$ is a YES-instance of \TDM iff $(N,\mathcal{T},V,t)$ is a YES-instance of \textsc{CTC-VM}. 

For the forward direction, let \( \mathcal{S}' \) denote the subset of $\mathcal{S}$ of size $n$ such that each element in $\mathcal{U}$ appears in exactly one triple in $\mathcal {S}'$. Now we will construct a spanning caterpillar whose tournament value is at least $2|\mathcal{S}|+2n+1$ as follows. Let $h$ belong to the backbone. For the sets in $\SS'$, let $M' \subseteq M$ denote the corresponding players. These players will be part of the backbone of the spanning caterpillar. Using \cref{prop:ham}, we can find a Hamiltonian path in the subgraph \( G[M'] \). We attach this Hamiltonian path after the player \( h \) in the backbone of the spanning caterpillar. This step is valid because player \( h\) beats every player in \( M \). Players in $L$ are attached to the player in $M'$ in the backbone that beats them (recall that $v_S$ beats $v_a$ iff $a \in S$). Players in \( M \) corresponding to sets not in $\SS'$ are attached to player \( h \) as leaves. Note that the value of the tournament is $2|\mathcal{S}|+2n+1$.

For the reverse direction, assume the tournament value is at least $2|\mathcal{S}|+2n+1$. Suppose, for the sake of contradiction, that $(\mathcal{U},\mathcal{S})$ is a NO-instance of \TDM. Consider a solution where \(k\) players in \(M\) beat exactly \(\ell\) players in \(L\). This implies that the maximum contribution to the tournament value from the players in \(L\) is only \(\ell\). Additionally, these \(k\)
 players can contribute at most \(k - 1\) points to the tournament value by playing among themselves. An additional 2 points can be achieved if one of these players loses to $h$, but only one player can benefit from this. For the remaining \(|\mathcal{S}| - k\) players in \(M\), the maximum contribution can be \(2(|\mathcal{S}| - k)\)
points if they all lose to $h$. 

Therefore, the total value that can be achieved is $2|\mathcal{S}| - k + 1 + \ell.$ For this to be at least $2|\mathcal{S}|+2n+1$, we need $\ell-k\geq 2n$. Since we assume that the given instance is NO instance, from Proposition~\ref{prop:3dminequality}, $\ell-k < 2n$, giving a contradiction.
\end{proof}

Next, we present an approximation algorithm for player-popularity-based \textsc{CTC-VM}, with an approximation factor that depends on the number of distinct popularity values.
\begin{theorem}
    There is a $\frac{1}{k-1}$-approximation algorithm for player-popularity-based \textsc{CTC-VM} where the range of the popularity values is a set of cardinality $k$. 
\end{theorem}

\begin{proof}
    Let the range of the popularity values be $\{v_i: i\in[k]\}$ where we have $v_1>v_2>\cdots>v_k$. Let $P_i$ be the set of players whose popularity is $v_i$ for $i\in[k]$. Without loss of generality, we can assume that $v_k=0$. Our algorithm is as follows.
\begin{enumerate}
\item For every $i\in[k-1]$, run \textsc{CTC-VM}, by setting the value of every player except those in $P_i$ to be $0$, using the algorithm described in \Cref{thm:pp-alg}.
\item  Among these $k-1$ tournaments, select the one with the highest value as the output. 
\end{enumerate}

    Our algorithm runs in polynomial time since the algorithm in \Cref{thm:pp-alg} runs in polynomial time. We next prove its approximation guarantee.

    Let us consider a challenge the champ tournament $T$ of the highest value for the input instance. Let its value be OPT. Let $V_i$ be the sum of the values of the matches in $T$ where a player from $P_i$ has won for $i\in[k-1]$. Then we have
    \[ OPT = \sum_{i=1}^{k-1} V_i. \]
    We now have the following if ALG is the value of the tournament output by the algorithm.
    \[ ALG \ge \max_{i=1}^{k-1} V_i \ge \frac{1}{k-1} \sum_{i=1}^{k-1} V_i = \frac{OPT}{k-1}. \]
\end{proof}



\section{Win-Count-Based Tournament Value Functions} \label{sec:wcb}

We now consider win-count-based value functions, modeling situations where a player's popularity changes through the course of the tournament. For these value functions, the value a player brings to a match depends on how many matches the player has won. We first prove that \textsc{CTC-VM-Dag} is polynomial-time solvable for win-count-based tournament value functions by giving a dynamic programming algorithm. We note that, in fact, Theorem~\ref{thm:pp-alg2} is a more general result than Theorem~\ref{thm:pp-alg1}. However, it does not give an explicit value for the optimal seeding. Theorem~\ref{thm:pp-alg1} gives an explicit seeding and value, as well as a simpler greedy algorithm. Further, as discussed earlier, it allows us to show that for DAGs and player-popularity-based value functions, Challenge the Champ tournaments are optimal among all single-elimination tournaments.

\begin{restatable}{theorem}{THMA}\label{thm:pp-alg2}
    There is a polynomial-time algorithm for  \textsc{CTC-VM-Dag} for win-count-based tournament value functions.
\end{restatable}
\begin{proof}
    We are given a set $N = \{1, 2, \ldots, n\}$ of players, a directed acyclic strength graph $\TT$, and win-value functions $f_i$ for each player. Since the strength graph $\mathcal{T}$ is a DAG, we define a linear ordering of the players based on their strengths. As before, we assume that the players are ordered as $n, n-1,\ldots, 1$, where player $n$ is the strongest and player 1 is the weakest.
    
    The win-value function $f_{i}(\ell)$ is the value contributed by player $i\in N$ to the tournament value for the $\ell$th match won. Let $F_i(\ell) = \sum_{k=1}^{\ell} f_i(k)$ be the total value obtained by player $i$ for $\ell$ wins. Our dynamic programming table has the entries of the form $\mathcal{A}(i,k)$, representing the maximum tournament value achievable by a Challenge the Champ tournament with $k$ matches and seeded by a subset of the players $\{i,i-1,\ldots 1\}$.

     For the base case, we set $\mathcal{A}(i,0)=0$ for all $i\in[n]$. This represents the scenario where no matches have been played, thus the tournament value is zero.
    
    For the recursive step, we need to compute $\mathcal{A}(i,k)$ for every pair $(i,k)$ where $i\in [n]$ and $k\leq i-1$ (as $i$ players can play at most $i-1$ matches). For this purpose, we use the following recursive relation: $$\mathcal{A}(i,k)=\max_{0\leq \ell\leq k}\{F_{i}(\ell)+\mathcal{A}(i-1,k-\ell)\}.$$

For player $i$, intuitively,  the recursive relation considers all possible ways the player 
$i$ can contribute to the total tournament value. Our aim is to compute $\mathcal{A}(n,n-1)$. 

Next, we show that the dynamic program is correct. For that purpose, we apply induction on 
$i$.
For the base case with $i=1$, observe that a single player cannot play any match; therefore, the tournament value is 0. For the inductive step, assume the algorithm correctly computes the maximum tournament value for all sets of players up to $i-1$ (with $k\leq i-1$). 

Let $\mathsf{OPT}(i,k)$ denote the optimal tournament value achievable using players $\{i,i-1,\ldots 1\}$ with exactly $k$ matches played. Furthermore, assume that $\mathsf{OPT}(i-1,k)=\mathcal{A}(i-1,k)$ for all values of $k\leq i-2$. 

First, we prove that $\mathsf{OPT}(i,k)\leq \mathcal{A}(i,k)$ for every $i\in [n]$ and $k\leq i-1$. Suppose in the optimal tournament seeding, player $i$ wins $\ell^{*}$ matches. Note that we assume player $i$ participates in at least one match, which means $\ell^* \geq 1$; otherwise, the proof would be trivial as $\mathsf{OPT}(i,k)$ will be equal to $\mathsf{OPT}(i-1,k)$ in that case. Furthermore, $\ell^* \geq 1$ implies that $k - \ell^* \leq i - 2$ (this is necessary because we want to use the induction hypothesis later). Next, by the recursive relation, we have $\mathcal{A}(i,k)\geq F_{i}(\ell^{*})+\mathcal{A}(i-1,k-\ell^{*})$. In the optimal solution, the value contributed by player $i$
is $F_{i}(\ell^{*})$ and the remaining value is achieved by the other players playing optimally in the remaining $k-\ell^{*}$ matches. Hence, $\mathsf{OPT}(i,k)=F_{i}(\ell^{*})+\mathsf{OPT}(i-1,k-\ell^{*})$. 
 Since  $\mathsf{OPT}(i-1,k-\ell^{*})=\mathcal{A}(i-1,k-\ell^{*})$ (by the inductive hypothesis), it follows that $\mathsf{OPT}(i,k)\leq \mathcal{A}(i,k)$.
 
 Next, we prove that  $\mathsf{OPT}(i,k)\geq \mathcal{A}(i,k)$ for every $i\in [n]$ and $k\leq i-1$.  In order to prove this, let  $\mathcal{A}(i,k)=F_{i}(\ell^{*})+\mathcal{A}(i-1,k-\ell^{*})$ for some $\ell^{*}\leq k$ such that $k-\ell^{*}\leq i-2$. 

Due to the induction hypothesis, we note that $\mathsf{OPT}(i-1,k-\ell^{*})=\mathcal{A}(i-1,k-\ell^{*})$. Note that when we use the optimal strategy for $i-1$ players playing $k-\ell^{*}$ matches, then at least $\ell^{*}$ players remain unbeaten $($as $i-1-(k-\ell^{*})\geq \ell^{*})$. Since player $i$ is stronger than all of these unbeaten players, if we allow player $i$ to play against all of these unbeaten players, the tournament value obtained by player $i$ winning $\ell^*$ matches is $F_i(\ell^*)$, and the remaining matches are played among the $i-1$ players optimally. Note that among these $\ell^{*}$ players that $i$ plays against, $\ell^{*}-1$ have not played any matches, and only one is the winner of the Challenge the Champ tournament conducted among the previous $i-1$ players. 
Then the value we obtain is $v_{i}(\ell^{*})+\mathcal{A}(i-1,k-\ell^{*})$ which is at most $\mathsf{OPT}(i,k)$. Thus, $\mathcal{A}(i,k) \leq \mathsf{OPT}(i,k)$ for every $i\in [n]$ and $k\leq i-1$.
\end{proof}

Next, we show that \textsc{CTC-VM} is \NPH for binary-valued win-count-based games. Recall that here, every player has a threshold $\lambda_i$, and the tournament value for a seeding $\sigma$ is $|\{i : w_i(\sigma) \ge \lambda_i\}|$. For this, we give a polynomial-time reduction from \textsc{Independent Set}. The input of \textsc{Independent Set} consists of an undirected graph $\GG=(\VV,\EE)$ and a positive integer $k \in \mathbb{N}$. The problem is to determine if there exists a subset $\WW \subseteq \VV$ of size $k$ so that every edge in $\EE$ is incident on at most one vertex in $\WW$.

\begin{restatable}{theorem}{THM} \label{thm:ccvm-nph-bin}
    \textsc{CTC-VM} is \NPH for binary-valued win-count-based tournament value functions.
\end{restatable}

\begin{proof}
    
    Let $(\GG = (\VV,\EE),k)$ be a given instance of \textsc{Independent Set}. Let $n = |\VV|$, and $d(v)$ be the degree of vertex $v \in \VV$. We construct an instance $(N,\TT,f,t)$ of \textsc{CTC-VM} as
    \begin{align*}
        N &= \{p_v: v\in\VV\} \cup \{p_e^\el: e\in\EE, \el\in[n^2]\} \, .
    \end{align*}

    Thus there is a player $p_v$ for every vertex in $\GG$ whom we call a \emph{vertex player}, and $n^2$ players $(p_e^\el)_{\el \in [n^2]}$ for every edge $e$ in $\EE$ whom we call \emph{edge players}. We now describe the strength graph $\TT$ between the players. If $v \in e$ in $\GG$, then the player $p_v$ beats $p_e^\el$, otherwise $p_e^\el$ beats $p_v$. Thus a vertex player beats all edge players for incident edges, and is beaten by all other edge players. The directions of the remaining edges of the strength graph are arbitrary. An illustration of the construction of a strength graph from a given instance of \textsc{Independent Set} is shown in \cref{fig2}. The value of a vertex player $p_v$ is $1$ if she wins at least $d(v) \times n^2$ matches; otherwise, her value is zero. Thus the threshold $\lambda_{p_v} = d(v) \times n^2$. The value of all edge players is zero, irrespective of the number of matches she wins. The target value $t = k$. We claim that there is a seeding that achieves the target value if and only if the graph $\GG$ has an independent set of size $k$.

    In one direction, let $\WW\subseteq\VV$ form an independent set of size $k$, and assume $\WW= \{v_1, \ldots, v_k\}$. Let $N_\WW$ be the corresponding set of players $\{p_v : v \in \WW\}$, and let $\HH_\WW$ be a Hamiltonian path on the players in $N_\WW$. By reindexing the vertices, assume $\HH_\WW = (p_{v_k}, p_{v_{k-1}}, \ldots, p_1)$. Note that this implies that player $p_{v_{i+1}}$ beats player $p_{v_i}$ in the strength graph, and by construction and since $\WW$ is an independent set, each edge player $p_e^\ell$ is beaten by at most one player in $N_\WW$. Consider the seeding that starts with player $p_{v_1}$, then consists of all $d(v_1) \times n^2$ edges players beaten by $p_{v_1}$, then has player $p_{v_2}$ followed by all edge players beaten by $p_{v_2}$. We continue in this manner until $p_{v_k}$ and all edge players are beaten by this player. The remaining players are then added in this sequence arbitrarily. Since each player $p_v$ for $v \in \WW$ beats $d(v) \times n^2$ edge players, it is easy to see that this seeding has the desired tournament value.
    
    \begin{figure}[t]
 \centering
    \includegraphics[scale=.65]{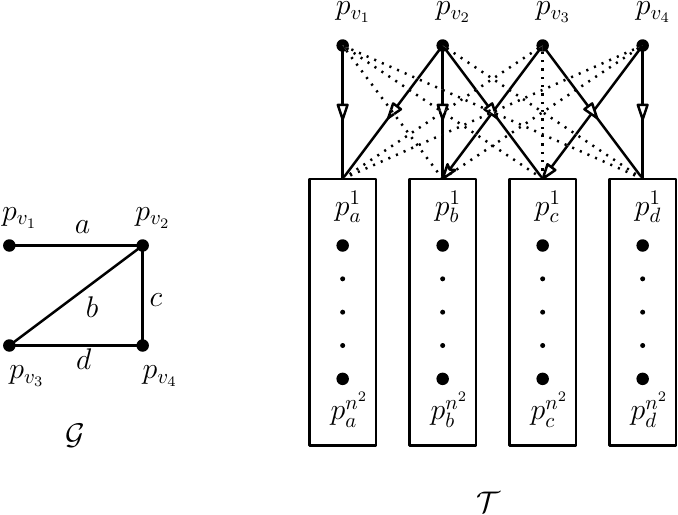}
    \caption{ An edge from a vertex to a box denotes all directed edges between that vertex and vertices within the box in the same direction. Edges not depicted are arbitrarily directed. The dotted edges are directed from the vertices in the box to the outside vertex.} 
    \label{fig2}
\end{figure}

    In the other direction, suppose there is a seeding with tournament value $k$. Clearly, there must be $k$ vertex players that have positive value, since only vertex players can have positive value. Let $N_\WW = \{p_{v_1}, \ldots, p_{v_k}\}$ be these vertex players, and $\WW = \{v_1, \ldots, v_k\}$ are vertices corresponding to these players in the graph $\GG$. We claim that the set $\WW$ must be an independent set in $\GG$. If not, suppose $v_i, v_j \in \WW$ are adjacent, with edge $\hat{e} = \{v_i,v_j\} \in \EE$. Then both $p_{v_i}$ and $p_{v_j}$ have outgoing edges to the $n^2$ edge players $p_{\hat{e}}^\ell$ in the strength graph. Consequently, at most $n + n^2 (d(v_i) + d(v_j) -1)$ players have incoming edges from either $v_i$ or $v_j$,\footnote{The additional $n$ players are due to vertex players possibly beaten by $p_{v_i}$ and $p_{v_j}$.} and hence for large enough $n$, $v_i$ and $v_j$ cannot each beat $d(v_i) \times n^2$ and $d(v_j) \times n^2$ players respectively. It follows immediately that $\GG$ has an independent set of size $k$, completing the proof.
\end{proof}

In the proof of \Cref{thm:ccvm-nph-bin}, the maximum value of any Challenge the Champ tournament obtained in the reduction is the same as the maximum size of an independent set in the \textsc{Independent Set} instance. Hence, we obtain the following from the known inapproximability of \textsc{Independent Set}~\cite{DBLP:conf/focs/Hastad96,DBLP:conf/focs/Khot01,DBLP:conf/stoc/Zuckerman06}.

\begin{corollary}
    For every real number $\eps>0$, no polynomial-time algorithm approximates the value of \textsc{CTC-VM} for binary-valued win-count-based tournament value functions within a factor of $n^{1-\eps}$ unless $\Pb=\NP$.
\end{corollary}

We next consider linear-after-threshold value functions. Recall that now each player $i$ has a threshold $\lambda_i$, and for a seeding $\sigma$, the value obtained is $\sum_{i \in N} \max\{0, w_i(\sigma) - \lambda_i + 1\}$. By a modification of the previous reduction, we show next that \textsc{CTC-VM} is \NPH even for linear-after-threshold valuation functions. However, unlike Theorem~\ref{thm:ccvm-nph-bin}, the reduction is not approximation-preserving, hence we do not obtain the same approximation hardness result.


\begin{restatable}{theorem}{THMB}\label{thm:ccvm-nph-lin}
    \textsc{CTC-VM} is \NPH for linear-after-threshold valuation functions.
\end{restatable}

\begin{proof}
    We again reduce from \textsc{Independent Set} to prove \NP-hardness. The reduced instance of \textsc{CTC-VM} is the same as in the proof of \Cref{thm:ccvm-nph-bin}, except for the valuation functions, which are as follows. Since the valuation functions are linear-after-threshold, it is enough to define the threshold of every player. The threshold $\lambda(p_v)$ of any vertex player $p_v$ is $(d(v)\times n^2)-2n$. The threshold of every edge player is $|N|$, the total number of players. Hence as before, only vertex players have positive values in any seeding. We claim that the optimal tournament has value at least $2nk$ if and only if $\GG$ has an independent set of size $k$. For the first direction, as in the proof of~\Cref{thm:ccvm-nph-bin}, an independent set of size $k$ gives us a seeding with $k$ vertex players, and each vertex player $p_v$ beats $d(v) \times n^2$ edge players corresponding to incident edges. Since the threshold is $\left((d(v)\times n^2\right)-2n$, each of these players gets value at least $2n$.

    For the other direction, given a seeding of value $2nk$, recall that only vertex players can obtain any value. Since there are $n$ player vertices, at most value $n$ can be obtained by beating other vertex players. The remaining $2nk-n$ value must be obtained by beating edge players. Each vertex player $p_v$ can beat at most $d(v) \times n^2$ edge players, and has threshold $\left(d(v)\times n^2\right)-2n$, and hence can obtain value at most $2n$. Thus at least $k$ vertex players must reach their threshold. Let $N_\WW = \{p_{v_1}, \ldots, p_{v_k}\}$ be these vertex players, and $\WW = \{v_1, \ldots, v_k\}$ are vertices corresponding to these players in the graph $\GG$. Following the earlier proof, we claim that the set $\WW$ must be an independent set in $\GG$. If not, suppose $v_i, v_j \in \WW$ are adjacent, with edge $\hat{e} = \{v_i,v_j\} \in \EE$. Then both $p_{v_i}$ and $p_{v_j}$ have outgoing edges to the $n^2$ edge players $p_{\hat{e}}^\ell$ in the strength graph. Consequently, at most $n + n^2 (d(v_i) + d(v_j) -1)$ players have incoming edges from either $v_i$ or $v_j$, and hence for large enough $n$, $v_i$ and $v_j$ cannot each beat $d(v_i) \times n^2 - 2n$ and $d(v_j) \times n^2 - 2n$ players respectively. It follows immediately that $\GG$ has an independent set of size $k$, completing the proof.
\end{proof}

\section{Pair-Based Tournament Value Functions} \label{sec:ggv}

For pair-based tournament value functions, we show that \textsc{CTC-VM-Dag} is \NPH, even for binary valuations. Our proof gives a reduction from \TDM. The reduction is similar to the proof of Theorem~\ref{thm:rb-np-com}, but the added flexibility of pair-based value functions allows us to slightly simplify the proof. 

\begin{restatable}{theorem}{THMC} \label{thm:dag:npc}
    \textsc{CTC-VM-Dag} is \textsf{NP}-complete for pair-based tournament value functions when the tournament value function maps to $\{0,1\}$.
\end{restatable}
\begin{proof}
For a given instance $(\UU,\SS)$ of \TDM with $|\UU| = 3n$ and $|\SS| = m$, the reduction creates a player $a_S$ for each set $S \in \SS$ (called \emph{set players}), a player $b_u$ for each element $u \in \UU$ (called \emph{element players}), and a new player $c$. There are thus $m +3n + 1$ players. In the strength graph $\TT$, $c$ beats all other players, and each $a_S$ beats all players $b_u$. All other edges in $\TT$ are directed arbitrarily. Each match between $c$ and any set player has value $1$. Each match between a set player $a_S$ and an element player $b_u$ has value 1 if $u \in S$ in the \TDM instance, otherwise it has value 0. All other matches have value 0. In particular, any match between $c$ and the element player, two element players, between two set players, or between set player $a_S$ and element player $b_u$ for $u \not \in S$ has value 0. The threshold $t = (3n+1) + (m-n) = m+2n+1$. See \cref{fig3} for an illustration.

We now show that the tournament value threshold $t$ is reached if and only the given \TDM instance is a YES instance. Suppose there exists a family of sets $\SS' \subseteq \SS$ of size $n$ that cover $\UU$. We construct the spanning caterpillar as follows. From Proposition~\ref{prop:ham}, let $H$ be a Hamiltonian path on the vertices $\{a_S\}_{S \in \SS'}$, starting from player $a^*$. Note that each element player $b_u$ is beaten by some player $a_S$ for $S \in \SS'$. The backbone consists of player $c$, followed by the Hamiltonian path $H$ starting from player $a^*$. For the leaves, note that $c$ defeats each set player. Then each set player $a_S$ for $S \not \in \SS'$ is attached as a leaf to player $c$. All element players are attached as leaves to the player $a_S$ in the backbone that beats them and gets value 1 in doing so. Thus each set player in the backbone beats 3 element players. Then in this seeding, player $c$ beats player $a^*$ at the head of $H$, and all players $a_S$ for $S \not \in \SS'$, for a total value of $1+m-n$. See \cref{fig3} for an illustration. Each of the $n$ set players in the backbone beats 3 element players and hence gets a total value of $3n$. Thus, if the given instance is a YES instance, the given seeding gets the value $3n + 1+m -n$ $=m+2n+1$, which is the threshold.

If the given instance is a NO instance, consider any spanning caterpillar. Since $c$ beats every other player, she must be at the head. Let $A'$ be the set players that beat at least one element player and contribute a value of 1 as a result, and let $B'$ be the element players beaten by a set player. Since set player $a_S$ can beat element player $b_u$ and simultaneously achieve a value of 1 only if $u \in S$ in the \TDM instance, from Proposition~\ref{prop:3dminequality}, $|B'| - |A'| < 2n$. The players in $A'$ must all be in the backbone. Then since the backbone must be a path, $c$ can beat at most one player from $A'$. Hence $c$ gets total value at most $1 + (m-|A'|)$. Further, the set players in the backbone beat $|B'|$ element players, and hence get value exactly $|B'|$. As observed, element players in the backbone do not contribute any value, and hence the value obtained from this seeding is $1 + m - |A'| + |B'|$ $< 1+m + 2n$ $=t$.
\end{proof}
\begin{figure}
 \centering
 \begin{minipage}{0.38\textwidth}
     \centering
     \includegraphics[width=\textwidth]{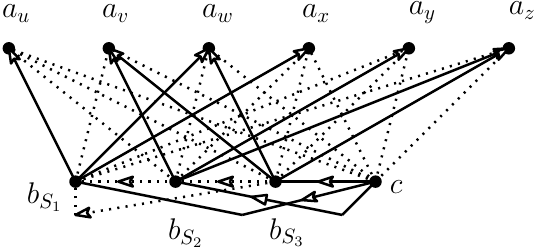}

 \end{minipage}
 \begin{minipage}{0.29\textwidth}
     \centering
     \includegraphics[width=\textwidth]{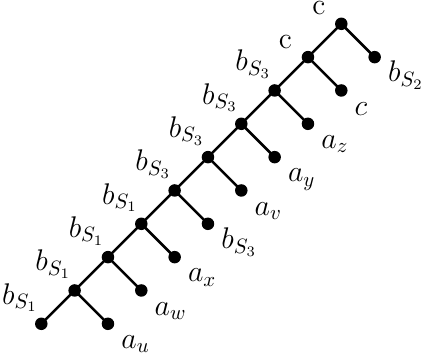}

 \end{minipage}
 
 \caption{In the left, we have the strength graph $\mathcal{T}$ of \textsc{CTC-VM-Dag} corresponding to \TDM instance $(\mathcal{U},\mathcal{S})$, where $\mathcal{U}=\{\{u,v\},\{w,y\},\{x,z\}\}$ and $\mathcal{S}=\{S_{1}=\{u,w,x\},S_{2}=\{v,w,z\},S_{3}=\{v,y,z\}\}$. Here, bold edges have value 1, non-depicted edges are directed arbitrarily with value 0, and dotted edges with no directions are directed from clause players to set players and have value 0. In the right, we have a spanning caterpillar corresponding to solution $\mathcal{S'}=\{S_{1},S_{3}\}$ of the \TDM.}
 \label{fig3}
\end{figure}

Consider the following problem. Given a complete directed acyclic graph $\GG$ with binary weights $w_e \in \{0,1\}$ on each directed edge, the problem is to find a maximum weight caterpillar arborescence --- i.e., an arborescence of maximum weight such that the removal of all vertices with out-degree zero in the arborescence gives a directed path. We can show this problem is also \NPH, since there is a direct reduction from \textsc{CTC-VM-Dag} with binary pair-based tournament value functions. We simply let $\GG$ be the strength graph $\TT$, and set the weight $w_e$ of an edge $e=(i,j)$ to be the corresponding value for the pair of players $i,j$. It follows from Theorem \ref{thm:dag:npc} that this problem is also \NPH.

\begin{restatable}{corollary}{CORR} \label{corr:sc}
  Computing a maximum weight caterpillar arborescence in a DAG is \NPH.  
\end{restatable}

\section{Conclusion and Open Problems} \label{sec:conc}

In this paper, we initiate the study of tournament value maximization for Challenge the Champ tournaments and contribute to the growing body of research on tournament fixing and design problems. Our results provide a complete and comprehensive picture of the computational complexity of value maximization, including for binary and ternary value functions. Theorems~\ref{thm:pp-alg1} and~\ref{thm:pp-alg}, in fact, provide upper bounds on value maximization for \emph{all} single-elimination tournaments. En route, we show interesting connections to Hamiltonian paths in complete directed graphs and spanning caterpillar arborescences. 

An obvious question for future research in value maximization is approximation. All of our hardness results are for binary or ternary valuations. These results hence show strong NP-hardness, and rule out fully polynomial time approximation schemes (FPTAS), unless P = NP. However, weaker approximations may be possible. We believe that given that value maximization is a practical concern, this is an interesting and important direction for research.

A second question is with regard to parameterized complexity. There are many relevant and natural parameters worthy of investigation. Our paper shows that for win-count-based tournament value functions, value-maximization is easy if the strength graph is a DAG but is NP-hard otherwise, even for binary values. An obvious parameter to try is then the size of the \emph{feedback vertex set} or the \emph{feedback arc set} of the strength graph. Also, for player-popularity-based games with ternary values, the problem is \NPH even with a single player with value 2. It is possible that using the number of 0 or 1 value players as a parameter would be helpful. 

Lastly, value maximization in other tournament formats, including knockout tournaments and extended stepladder tournaments, has only received limited attention and presents a rich and important avenue for further research.

   \bibliographystyle{abbrvnat}
   \bibliography{references}

\end{document}